\documentclass[copyright,creativecommons]{eptcs}
\providecommand{\event}{ACAC 2009} 
\providecommand{\volume}{??}
\providecommand{\anno}{2009}
\providecommand{\firstpage}{1}
\providecommand{\eid}{??}
\usepackage{breakurl}        
\usepackage{amssymb}
\usepackage{graphicx}
\usepackage{amsmath}
\usepackage{ graphics}
\usepackage{ epsfig}
\usepackage{latexsym}
\usepackage{amssymb}

\newcommand{\var}{\mbox{\sf Var}}
\newcommand{\B}{\mbox{\sf B}}
\newcommand{\cls}{\mbox{\sf Cls}}

\newcommand{\php}{\mbox{\sf PHP}}
\newcommand{\nc}{\mbox{\sf NC}}
\newcommand{\pc}{\mbox{\sf PC}}

\newcommand{\fa}{\mbox{\sf false}}
\newcommand{\tr}{\mbox{\sf true}}
\newcommand{\size}{\mbox{\sf size}}

\newtheorem{theorem}{Theorem}[section]
\newtheorem{lemma}[theorem]{Lemma}
\newtheorem{definition}[theorem]{Definition}

\newenvironment{proof}[1][Proof]{\begin{trivlist}
\item[\hskip \labelsep {\bfseries #1}]}{\end{trivlist}}

\title{An Exponential Lower Bound on OBDD Refutations for Pigeonhole Formulas}
\author{Olga Tveretina\thanks{This work
      was supported in part by the ``Concept for the Future'' of  
Karlsruhe Institute of
      Technology within the framework of the German Excellence  
Initiative.}
\institute{Institute for Theoretical Computer Science\\
Karlsruhe University\\
Am Fasanengarten 5, D-76131 Karlsruhe\\ Germany\\}
\email{olga@ira.uka.de}
\and
Carsten Sinz\footnotemark[\value{footnote}]
\institute{Institute for Theoretical Computer Science\\  
Karlsruhe University\\
Am Fasanengarten 5, D-76131 Karlsruhe\\ Germany\\}
\email{\quad sinz@ira.uka.de }
\and
Hans Zantema
\institute{Department of Computer Science\\  TU Eindhoven,
 The Netherlands}
\institute{Institute for Computing and Information Sciences\\
 Radboud University,
The Netherlands}
\email{\quad h.zantema@tue.nl }
}
\def\titlerunning{An Exponential Lower Bound on OBDD Refutations for Pigeonhole Formulas}
\def\authorrunning{O. Tveretina \& C. Sinz \& H. Zantema}
\begin{document}

\date{ }

\maketitle

\begin{abstract}  Haken proved that every resolution refutation of the pigeonhole formula has at least exponential size.
Groote and Zantema proved that a particular OBDD computation of the pigeonhole formula has an exponential
size. Here we show that any arbitrary OBDD refutation of the pigeonhole formula  has an exponential
size, too: we prove that the size of one of the intermediate OBDDs is  $\Omega(1.025^n)$.

\end{abstract}

\section{Introduction}

The pigeonhole principle, also known as Dirichlet's box  principle states that $n$ holes can hold at most $n$ objects with one object to a hole.
The propositional formulas describing this principle
 were introduced by Cook and  Reckhow in 1979 \cite{CR1979}.
The formula is a CNF parameterized by $n$.  It is unsatisfiable, but after removing any single clause it becomes satisfiable, it is thus minimally unsatisfiable.

The formula  has a very simple shape, a meta argument for unsatisfiability is easily given, but standard techniques for proving
unsatisfiability automatically  run out of time for quite small
values of $n$. 
Therefore, this formula is a good benchmark to   test  the efficiency of an approach for deciding (un)satisfiability.

Also, on the theoretical side, it is the basis of many interesting results.
A landmark result is that of Haken \cite{H1985}, who proved that the length of
any resolution refutation of the pigeon hole formula is at least exponential
in $n$. Surprisingly, Cook proved that it admits a
polynomial refutation based on extended resolution \cite{C1976}.

An   {\it Ordered Binary Decision Diagram} (OBDD), also referred as a reduced OBDD (ROBDD) or just a BDD, is a data structure that is used to represent Boolean functions \cite{B1986,W2000}.

 OBDDs have some interesting properties: they provide compact and canonic representations of Boolean functions, and there are efficient algorithms for performing logical operations on OBDDs.  As a result,
OBDDs have been successfully applied to a wide variety of tasks, particularly
 in VLSI design and CAD  verification \cite{MT1998}. There are some less well-known  applications as  fault tree analysis \cite{SA1996}, Bayesian reasoning and product configuration.

 As a propositional proof system OBDDs  were studied, e.g., by Atserias {\em et al.} \cite{AKV2004}. The authors introduce a very general proof system based on constraint propagation. OBDDs are a special case of this proof system. Their proof system has four rules: {\it axiom}, {\it join}, {\it projection}, and {\it weakening}.  The first two rules, {\it axiom} and {\it join},  correspond to an application of the  OBDD {\it apply} operator. {\it  Projection} and {\it weakening} are introduced to reduce the size of intermediate OBDDs. It was shown that the OBDD proof system containing  all four rules is strictly stronger than resolution \cite{AKV2004} but it is still exponential \cite{K2008}.

In our paper, by the OBDD proof of a formula $\varphi$ we  mean the computation of the corresponding OBDD using the {\it apply-}operation, i.e. in terms of  the above proof system from \cite{AKV2004}, we allow only two rules, namely {\it axiom} and {\it join}.
 If the formula contains $n$ Boolean connectives, then the OBDD construction requires exactly $n$ calls of  ${\it apply}$, and the exponential blow up of the size of the proof is caused by the expansion of the size of the arguments. 

 In
\cite{GZ2003} it was proved that a particular OBDD computation of the pigeonhole formula is at least exponential.   On the other hand,  it was proved in \cite{CZ2009} that the pigeonhole formula admits a polynomial size OBDD refutation in a setting including existential quantification (i.e. including the {\it projection} rule).

In this paper we prove that, based on the notion of OBDD refutation along the
lines of \cite{CZ2009} containing the classical ingredients of OBDD computation,
but excluding existential quantification, we have an exponential lower bound
for the size of OBDD refutations of the pigeonhole formula. This is much
stronger than the result from \cite{GZ2003}: there, the  only  computation
considered first computes the conjunction of all positive clauses, then
the conjunction of all negative clauses, and finally the conjunction of these two. In our setting, the clauses of the pigeonhole formula may be processed in any arbitrary order. We show that in any OBDD refutation proof  some of the intermediate OBDDs has size  at least exponential in $n$. As a consequence we state that the gap between polynomial and exponential in the OBDD refutation framework for pigeonhole formula  is caused by the rule for existential quantification.

We start with preliminaries in Section \ref{prel}. In Section \ref{main_result} we prove an exponential lower bound on OBDD refutations  for the pigeonhole formula. Finally,  Section \ref{conclusion} contains conclusions.

\section{Preliminaries}\label{prel}

We consider propositional formulas in {\it Conjunctive Normal Form} (CNFs). Basic blocks for building CNFs are propositional variables that take the values {\fa} or {\tr}. The set of propositional variables is denoted by $\var$. A literal is either a variable $x$ or its negation $\lnot x$.  A clause is a disjunction of literals, and a CNF is a conjunction of clauses. In the following,  for convenience, we consider clauses as sets of variables, and a CNF as a set of clauses.  By $\cls(\varphi)$  we denote the set of clauses contained in a CNF $\varphi$ and by $\var(\varphi)$ we denote the set of variables contained in the CNF $\varphi$.

\subsection{Ordered Binary Decision Diagrams}

An Ordered Binary Decision Diagram (OBDD) is  a a rooted, directed, acyclic graph, which consists of decision nodes and two terminal nodes 0 and 1.
Each decision node is labeled by a propositional variable from $\var$ and has two child nodes called low child and high child. The edge from a node to a low (high) child represents an assignment of the variable to 0 (1). Such a structure is called {\it ordered} because different variables appear in the same order on all paths from the root. Therefore,  OBDDs assume that  there is a total order $\prec$ on the set of variables $\var$.

 A OBDD is said to be {\it reduced} if the following two rules have been applied to its graph: 1) merge  isomorphic subgraphs; 2) eliminate any node whose two children are isomorphic. In our paper we consider only reduced OBDDs.


Given a propositional formula $\varphi$ and an order on variables $\prec$, we define the size of a OBDD  $\B(\varphi,\prec)$ representing $\varphi$ with respect to  $\prec$ as  the number of its internal nodes and denote it by $\size(\B(\varphi,\prec))$.

We give a definition of a OBDD refutation adapting the definition from \cite{CZ2009}.

\begin{definition}[OBDD refutation]  Given a total  order on variables $\prec$,  a OBDD refutation  of an unsatisfiable CNF $\varphi$ is a sequence of OBDDs $\B_1(\varphi_1,\prec), \dots, \B_n(\varphi_n,\prec)$  such that $\B_n(\varphi_n,\prec)$ is a OBDD representing the constant $\fa$ and for each $\B_i(\varphi_i,\prec)$, $1\leq i\leq n$, exactly one of the following holds.
\begin{itemize}
\item (Axiom) $\B_i(\varphi_i,\prec)$ represents one of the clauses $C\in\varphi$;
\item (Join)  there are OBDDs $\B_{i'}(\varphi_{i'},\prec)$ and $\B_{i''}(\varphi_{i''},\prec)$ such that $1 \leq i'<i''<i$  and $\varphi_i=\varphi_{i'}\wedge \varphi_{i''}$.
\end{itemize}
We say that $n$ is the length of the OBDD refutation. The size of the OBDD refutation is defined as $\sum_{i=1}^n \size(\B_i(\varphi_i,\prec))$.
\end{definition}

When it is convenient, instead of $\B(\varphi,\prec)$ we write $\B(\varphi)$ or just $\B$. If a OBDD $\B$ represents a CNF  $\varphi$ then by $\cls(\B)$ we mean $\cls(\varphi)$  and by  $\var(\B)$   we mean   $\var(\varphi)$.

 The size of the minimal OBDD representing a propositional formula $\varphi$ for a given order on variables $\prec$   is described by the following structure theorem  \cite{SW93,GZ2003}.  We use $\mathbb {B}=\{0,1\}$ to denote the set of Boolean constants.

\begin{theorem}\label{BddSize}  Suppose for a given formula $\varphi$  the following holds:
\begin{itemize}
\item $|\var(\varphi)|=n$;
 \item $\prec$ is a total order on the set of  variables $\var(\varphi)$;
\item $x_1, \dots, x_k$ are the smallest $k$ elements with respect to $\prec$ for some $k<n$;
\item $A\subseteq \{1,\dots,k\}$;
\item $z = (z^1,\dots,z^k) \in \mathbb{B}^k$.
\item For all distinct $\overrightarrow{x}_1, \overrightarrow{x}_2\in \mathbb{B}^k$ such that $x_1^i=x_2^i=z^i$ for all $i\not\in A$ there exists a $\overrightarrow y\in \mathbb {B}^{n-k}$ such that $\varphi(\overrightarrow x_1,\overrightarrow y) \neq\varphi(\overrightarrow x_2,\overrightarrow y)$.
\end{itemize}
Then the size of the OBDD $\B(\varphi,\prec)$ is at least $2^{|A|}$.
\end{theorem}

The proof of the lower bound presented in Section \ref{main} is based on Theorem \ref{BddSize}.  However, in order to obtain a lower bound we still have to solve some combinatorial problems. 

\subsection{The pigeonhole formula}

The pigeonhole principle states that $n$ holes can hold at most n objects with one object in a hole. It can be formulated as a set of clauses as follows.

\[ \pc_n = \bigwedge_{i=1}^{n+1} (\bigvee_{j=1}^{n} P_{ij}), ~~ \nc_n  = \bigwedge_{\substack{1\leq i<j \leq n+1 \\ 1\leq k \leq n}} (\lnot P_{ik} \vee \lnot P_{jk}) \]

\[\php_n= \pc_n \wedge \nc_n\]

Now we introduce notations that will be used in the rest of the paper.
Let \[\pc_n^*=\bigwedge_{i=1}^{n} (\bigvee_{j=1}^{n} P_{ij}) \enspace.\]
  Hence, $\pc_n^*$ contains the first $n$ clauses of $\pc_n$.
We represent $\pc_n^*$ as a matrix of variables with $n$ rows and $n$ columns (the clause $\bigvee_{j=1}^{n} P_{ij}$ corresponds to the $i$-th  row). 
We denote this matrix by $P$.
%
For each row in $P$ there is a corresponding clause in $\pc_n^*$ and vice versa, therefore we will refer to a  row  as a  clause,  and to a set of rows as a set of clauses.

For a given total order on variables $\prec$, we define   $S_{\prec}$ as the set  containing the $\lfloor n^2/2\rfloor$ smallest elements of $\var(\pc_n^*)$ with respect to ordering $\prec$, and let $S_{\succeq}=\var(\pc_n^*)\backslash S_{\prec}$.
Moreover, we define 
  \[S_{\prec}^*=\{ P_{ij}\in \var (\php_n) ~|~P_{ij}\preceq \max S_\prec\}, \] 
   and
\[S_{\succeq}^*=\var(\php_n)\backslash S_{\prec}^*.\]

 Note that  $S_{\prec} \cup S_{\succeq}=\var(\pc^*_n)$  and $S_{\prec}^* \cup S_{\succeq}^*= \var(\php_n)$. The sets $S_{\prec}$ and $S_{\succeq}$   are defined in such a way that the  difference between the sizes of these  sets is at most one, but, in contrary, this does not hold for the sets $S_{\prec}^*$  and  $S_{\succeq}^*$.

For each OBDD $\B_i$ in a OBDD    refutation of $\php_n$ we define
 \begin{center} $S_{\prec}^i=S_{\prec}^* \cap \var(\B_i)$ and $S_{\succeq}^i=\var(\B_i)\backslash S_{\preceq}^*.$ \end{center}

\noindent Moreover, we define
\begin{center} $\cls^{neg}(\B_i)=\cls(\B_i)\cap \cls(\nc_n)$ and $\cls^{pos}(\B_i)=\cls(\B_i)\cap \cls(\pc_n).$\end{center}


\section{The main result}\label{main_result}

The proof of our lower bound is inspired by the proof of  a lower bound of a particular OBDD refutation given in \cite{GZ2003}.

\begin{lemma}\label{matrix} 
Consider a matrix $M=\{m_{ij}\}$,   $1\leq i\leq
n$,  $1\leq j\leq n$. Let the matrix entries be colored  equally white and
black, i.e. the difference between the number of white entries and the number
of black entries is at most one. Let $m = \lfloor cn \rfloor$ for 
$c = \frac{1}{2} - \frac{1}{4} \sqrt{2} \approx 0.146$.
Then at least one of the following holds.
\begin{itemize}
 \item One can choose $m$ rows, and in every of these rows a white and a 
black entry, such that all these $2m$ entries are in different columns. 
 \item One can choose $m$ columns, and in every of these columns a white and a 
black entry, such that all these $2m$ entries are in different rows. 
\end{itemize}
\end{lemma}
\begin{proof}
Starting by the given matrix repeat the following process as long as possible. 
\begin{quote}
Choose a row in the matrix containing both a white and a black entry.
Remove both the column containing the white entry and the column containing 
the black entry. Also remove the chosen row.
\end{quote}
Assume this repetition stops after $k$ steps. If $k \geq m$ the first
property of the lemma holds and we are done. In the remaining case the
remaining matrix consists of $n-k$ rows with  $n-2k$ entries in each row, where every row
either only consists of white entries or only of black entries.
Assume that at least $n-2m$ of these rows are totally black. Using $k<m$
we conclude that the number of black entries in this remaining matrix is at least 
\[ (n-2m)(n-2k) > (n-2m)^2 \geq \frac{1}{2} n^2,\]
contradicting the assumption that at most half of the entries are black
(possibly up to one). So at least $n-k-(n-2m) = 2m-k > m$ of these rows are
totally white. By symmetry also at least $m$ of these rows are totally white. 
As the length of these rows are $n-k > n-m > m$, the second property of the lemma
is easily fulfilled. 
\end{proof}

By fine-tuning the argument the constant $c$ in Lemma \ref{matrix} can be
improved. We conjecture that it also holds for $c = 1 - \frac{1}{2}
\sqrt{2} \approx 0.293$. Choosing the $n \times n$ matrix in which the left
upper $k \times k$-square is black for $k \approx \frac{n}{\sqrt{2}}$ and the
rest is white, one observes that this value will be sharp. As our main result
involves an exponential lower bound, we do not focus on the precise optimal
value of $c$.

The pigeonhole formula is an unsatisfiable CNF and, hence, the OBDD representing $\php_n$ is just a terminal node $0$. Therefore, we have to show that for an arbitrary  order on variables and an arbitrary way to combine clauses there is an intermediate OBDD of a size exponential in $n$. 
We start our proof by the simple observations describing some properties of intermediate OBDDs.  And  the following lemma generalizes a well-known fact about binary trees claiming the existence of subtrees with
a weight lying between a and 2a (for any definition of ``weight'' as a sum of the weights of its leaves).

\begin{lemma}\label{general}
  Let $C$ be a finite set, $R \subseteq C$ with $|R| \geq 2$, and $B_1,\dots,B_l \subseteq C$
  a sequence with:
  \begin{enumerate}
  \item $B_l = C$
  \item For each $B_i$ ($1 \leq i \leq l$),
    either $B_i = \emptyset$, $B_i = \{ c \}$ for $c \in C$, or $B_i = B_j \cup B_k$ for some $j, k$ with $j < k < i$.
  \end{enumerate}
  Then, for each $a$ with $\frac{1}{|R|} < a \leq \frac{1}{2}$, there is a $j < l$ such that
  \[ a |R| \leq |B_j \cap R| < 2 a |R| \enspace .\]
\end{lemma}
\begin{proof}
We give a proof by contradiction. Suppose, for each $B_j$, either
 $$|B_j \cap R| < a |R| \qquad \textrm{or} \qquad |B_j \cap R| \geq 2 a |R| \enspace.$$

As $B_l \cap R = C \cap R = R$, the inequality $|B_l \cap R| \geq 2 a |R|$ holds for the final element
$B_l$ of the sequence. On the other hand, for singletons $B_j = \{ c \}$, we have
$|B_j \cap R| = 0 < a |R|$ for $c \notin R$, and $|B_j \cap R| = 1 < a |R|$ for $c \in R$, as $a > 1/|R|$.
Moreover, for $B_i = \emptyset$, $|B_i \cap R|  < a |R|$ obviously holds.
Following now the predecessors of $B_l$ (via the construction by set union)
in the sequence $B_i$ backwards, we finally arrive at an index $k$
for which the following holds:
\begin{itemize}
  \item $|B_k \cap R| \geq 2 a |R|$, and
  \item $B_k = B_{k'} \cup B_{k''}$, where $|B_{k'} \cap R| < a |R|$ and $|B_{k''} \cap R| < a |R|$.
\end{itemize}
As $B_k \cap R = (B_{k'} \cup B_{k''}) \cap R = 
(B_{k'} \cap R) \cup (B_{k''}  \cap R)$, and thus $|B_k \cap R| \leq |B_{k'} \cap R| + 
  |B_{k''} \cap R| < 2 a |R|$,  we arrive at a contradiction to $|B_k \cap R| \geq 2 a |R|$.
\end{proof}

\begin{lemma}\label{rows}

Suppose   $\B_1, \dots, \B_l$ is  a BDD refutation of  $\php_n$ and  $R\subseteq \cls(\pc_n)$ with
$|R| > 4$.
 Then   there is an $i < l$ such that  \[|R|/4 \leq  |\cls(B_i) \cap R| < 2 |R|/4 \enspace .\]

\end{lemma}

\begin{proof} Follows from Lemma \ref{general}.
\end{proof}

\bigskip

Let $\B_1, \dots, \B_l$ is  a BDD refutation of  $\php_n$. For each $i \leq l$ define $J_i$ as the set of
columns from $P^c$ as follows:
 \[J_i = \{ j\in \{1,\dots,n\} ~| ~\exists a,b: \lnot P_{aj}\vee \lnot P_{bj}\in \cls(B_i), ~ P_{aj}\in S_\prec, ~\text{and}~ P_{bj}\in S_{\succeq}\}.\]

\begin{lemma}\label{columns}

 Suppose  $\B_1, \dots, \B_l$ is  a BDD refutation of  $\php_n$ for a total order on variables $\prec$, and  $P'\subseteq \{1,\dots,n\}$ with $|P'| > 4$. Then  there is  an $i < l$  such that

\[ |P'|/4  \leq  |J_i \cap P'|  < |P'|/2.\]
 \end{lemma}

\begin{proof}
Follows from Lemma~\ref{general}, using $C = \{1,\dots,n\}$, $R = P'$, $a = 1/4$, and $J_1,\dots,J_l$ for the sequence
$(B_i)_{1 \leq i \leq l}$, for which the precondition of Lemma~\ref{general} holds, as is easily checked.
\end{proof}


\begin{theorem}\label{main}
For every order $\prec$  on the set of variables, the size  of each OBDD refutation  of $\php_n$ is
$\Omega(1.025^n)$.
\end{theorem}

\begin{proof}
Let $n>34$, and $\B_1, \dots, \B_l$ be a OBDD refutation  of $\php_n$.
We prove that for an arbitrary total order on variables $\prec$ there is an $i\leq l$ such that $\size (\B_i) \geq 2^{n(\frac{1}{2} - \frac{1}{4} \sqrt{2})/4}$.
 Since $2^{(\frac{1}{2} - \frac{1}{4} \sqrt{2})/4 }>1.025$ we have
$\size (\B_i) >1.025^n$ and the theorem holds.

We apply Lemma \ref{matrix} to the matrix representing $\pc_n^*$.  Then  one of the following holds.

\begin{itemize}
 \item There is a set  of $\lfloor n(\frac{1}{2} - \frac{1}{4} \sqrt{2})\rfloor $ rows (we denote this set by $R$) and there is a set of $2\lfloor n(\frac{1}{2} - \frac{1}{4} \sqrt{2})\rfloor $  entries (we denote this set  by $S^R$)  such that the following holds:
\begin{itemize}

\item  For each $r\in R$ there are $P_{ra}, P_{rb}\in S^R$ such that $P_{ra}\in S_{\prec}$ and $P_{rb}\in S_{\succeq}$.

\item  For distinct $P_{ab}, P_{cd}\in S^R$, $b\neq d$. \end{itemize}
We define  \[R^i= \cls (B_i) \cap R \enspace.\]

As $n>34$, $|R| = \lfloor n(\frac{1}{2} - \frac{1}{4} \sqrt{2})\rfloor \geq 5$, and we can apply
Lemma \ref{rows}. Thus we know that there is an $i< l$ such that
\[ |R|/4 \leq |R^i|< 2 |R|/4.\]

We get \[ 2|R^i|+1 \leq |R|.\]

 For each row $r\in R^i$ we fix an entry that is in the set $S_{\prec}$. We collect these elements in the set $A$. For each row $r\in R^i$ we also fix  an entry that is in $S_{\succeq}$ and collect these elements in the set $Y$. Let

\[R^j=\{j ~|~ \exists i: P_{ij}\in A \cup Y\}.\]


 Taking into account that $ 2|R^i|+1 \leq |R|$ we compute

\[|\cls^{pos}(B_i)|\leq (n+1) - (|R| -|R^i|) \leq (n+1) - ((2|R^i|+1) - |R^i|)= n- |R^i|.\]

We denote $\overline{R^i}=\cls^{pos}(B_i)\backslash R^i$. By definition $R^i\subseteq \cls^{pos}(B_i)$. Hence,  we obtain
\[|\overline{R^i}|=|\cls^{pos}(B_i)|-|R^i|\leq n- 2|R^i|.\]

Let  $J=n- |R^j|$. Since we have chosen the set of rows $R^i$ as satisfying the conditions of Lemma \ref{matrix}, we get $|R^j|=2|R^i|$ and

\[J = n - 2|R^i|\] and \[|\overline{R^i}|\leq |J|.\]

For each $C\in \overline{R^i}$ we fix one variable and collect these variables in the set $X$ that the following holds. For distinct $P_{ab},P_{cd}\in X$, $b\neq d$. This is possible because
$|\overline{R^i}|\leq |J|$.

We define $X_{\prec}=S_{\prec}^*\cap X$ and $X_{\succeq}=S_{\succeq}^*\cap X$.

We apply Lemma \ref{BddSize} on \[k=|S_{\prec}^i|.\]

For $j=1, \dots, k$ we define $z_j=1$ if $z_j\in A$ or $z_j\in X_{\prec}$, otherwise  we define  $z_j=0$.

Choose $\overrightarrow x, \overrightarrow x'$ satisfying $\overrightarrow x\neq \overrightarrow x'$ and $ x_j=x'_j=z_j$ for all $z_j\not\in A$. Then there is $j'$ such that $x_{j'}\neq x'_{j'}$.

Let $\overrightarrow y=(y_{k+1}, \dots, y_q)$, where $q=|\var(B_i)|$, be the vector defined  by $y_j=1$ if $y_j\in X_{\succeq}$ and  $y_j=0$ for all  $y_j\in S_{\succeq}^i\backslash (Y\cup X_{\succeq})$.  If $y_j\in Y$ then we choose $y_j=0$ if it is in the same row as $x_i$ and $y_j=1$ otherwise.

Hence, the subset of clauses represented by $\B_i$ evaluates to $x_{j'}$ for the assignment $ (\overrightarrow x, \overrightarrow y)$ and to $x'_{j'}$ for the assignment  $ (\overrightarrow x', \overrightarrow y)$.


The size of the set $A$ is at least $n ( \frac{1}{2} - \frac{1}{4} \sqrt{2} )/4$ by construction.  Hence, by Lemma \ref{BddSize}, we conclude that $\size (B_i)\geq 2^{|A|}\geq 2^{|R|/4}\geq 2^{n(\frac{1}{2} - \frac{1}{4} \sqrt{2})/4}$
for sufficiently large $n$.



\item There is a set of  $\lfloor n(\frac{1}{2} - \frac{1}{4} \sqrt{2})\rfloor $ columns (we denote this set by  $Q$) and there is a set  containing $2\lfloor n(\frac{1}{2} - \frac{1}{4} \sqrt{2})\rfloor$ entries (we denote this set by $S^Q$) such that the following holds:
 \begin{itemize} 
\item For each $q\in Q$ there are $P_{aq},P_{bq}\in S^Q$ such that $P_{aq}\in S_{\prec}$ and  $P_{bq}\in S_{\succeq}$. 
\item  For distinct $P_{ab},P_{cd}\in S^Q$, $a\neq c$.\end{itemize}
Suppose  $m=\lfloor n( \frac{1}{2} - \frac{1}{4} \sqrt{2})\rfloor$.

Let \[Q^c=\{ j~| ~\exists a,b: \lnot P_{aj}\vee \lnot P_{bj}\in \cls(B_i) ~\&~ P_{aj}\in S_{\prec} ~\&~ P_{bj}\in S_{\succeq}\}.\]

Then, by Lemma \ref{columns}, there is  $\B_i$ for $i<l$ such that
\[ m/4  \leq |Q^c| < m/2.\]
For each  $j\in Q^c$ we choose $\lnot P_{aj}\vee\lnot P_{bj}$ such that $\lnot P_{aj}\vee\lnot P_{bj}\in \cls(\B_i)$, where   $P_{aj}\in S_{\prec}$ and $P_{bj}\in S_{\succeq}$. We collect  $P_{aj}$ in $A$ and $P_{bj}$ in $Y$.

Let \[Q^r=\{a ~|~ \exists j: P_{aj}\in A \cup Y\}.\]

Let \[\overline {Q^c}=Q\backslash Q^c.\]
Then \[\overline {Q^c}>m/2.\]

For each $j\in \overline {Q^c}$ we fix $P_{a_{j}j}, P_{b_jj}\in S^Q$, where $P_{a_jj}\in S_{\prec}^*$ and $P_{b_jj}\in S_{\succeq}^*$. We collect $P_{a_{j}j}$ in $X_{\prec}$ and we collect   $P_{b_jj}$ in $X_{\succeq}$ for all $j\in \overline {Q^c}$.

We define
  \[\overline{Q^r}=\{a ~|~ \exists b: P_{ab}\in X_{\prec} \cup X_{\succeq}\}.\]
By Lemma \ref{matrix} all entries collected in $\overline{Q^r}$ are from different rows. Hence,  we obtain     \[|\overline {Q^r}|=2|\overline {Q^c}|.\]

Taking into account that $\overline {Q^c}>m/2$ we get
 \[\overline {Q^r}>2m/2=m\] and  since $\overline {Q^r}$ is a natural number we get
\[\overline {Q^r}\geq m+1.\]

We denote \[Q^*=\cls^{pos}(B_i)\backslash \overline {Q^r}.\]

The set of clauses  $\cls^{pos}(B_i)$ can contain an arbitrary subset of clauses from $\pc^n$, i.e. \[1 \leq |\cls^{pos}(B_i)|\leq n+1.\]  We take into account that  $|\overline{Q^r}|\geq m+1$ and compute
\[|\cls^{pos}|\leq (n+1)- |\overline{Q^r}|\leq  (n+1)- (m+1)=n-m.\]

We define $J=\{j ~| \exists a: P_{aj}\in \var(\php_n)~\&~j\not\in Q \}$. Then
 \[|J|=n-|Q|=n-m.\]

Therefore, $|Q^*|\leq |J|$.

For each row $r\in Q^*$ we fix one entry and  collect these entries in the set $W$. We require that the entries collected in $X$ satisfy the following properties.
\begin{itemize}
\item $r$ contains at least one entry such that this entry is in one of the columns of $J$;
\item each column is $J$ contains at most one fixed entry.
\end{itemize}
Since $|Q^*|\leq |J|$, there is such a set $W$. We denote $X_{\prec}^i=S_{\prec}^i \cap X_{\prec}$; $X_{\succeq}^i=S_{\succeq}^i \cap X_{\succeq}$;
$W_{\prec}=S_{\prec}^i \cap W$ and  $ W_{\succeq}=S_{\succeq}^i \cap W$.
We apply Lemma \ref{BddSize} on \[k=|S_{\prec}^i|.\]

For $j=1, \dots, k$ we define $z_j=1$ if $z_j\in A\cup X_{\prec}^i\cup W_{\prec}$, and we define  $z_j=0$ in all other cases.
We choose $\overrightarrow x, \overrightarrow x'$ satisfying $\overrightarrow x\neq \overrightarrow x'$ and $ x_j=x'_j=z_j$ for all $z_j\not\in A$. Then there is $j'\not\in\{1,\dots,k\}$ such that $x_{j'}\neq x'_{j'}$.
Let \[\overrightarrow y=(y_{k+1}, \dots, y_q),\] where $q=|\var(B_i)|$, be the vector defined  by $y_j=1$ for all $y_j\in X_{\succeq}^i$, $y_j\in W_{\succeq}$. 
For $y_j\in Y$ we define $y_j=1$ if  it  is in the same column as $x_{j'}$ and $y_j=0$ otherwise.   We choose  $y_j=0$ in all other cases. 
Therefore, for each row there is an entry that is assigned to 1 and for each column  except $j'$ and columns from the set $\overline {Q^c}$  there is at most one entry assigned to 1. If a column $t$  is contained in the set $\overline {Q^c}$ then two entries in this column can be assigned to 1. By construction, for each column $t$ in the set $\overline {Q^c}$  there is a clause  $\lnot P_{s't}\vee \lnot P_{s''t}\not\in \cls(\B_i)$. Therefore, assigning  $P_{s't}$ and $\lnot P_{s''t}$ simultaniously to 1 does not violate the satisfiability of the subformula represented by  $\B_i$. 

Hence, the subset of clauses represented by $\B_i$ evaluates to $x_{j'}$ for the assignment $ (\overrightarrow x, \overrightarrow y)$ and to $x'_{j'}$ for the assignment  $ (\overrightarrow x', \overrightarrow y)$.


The size of the set $A$ is at least $n ( \frac{1}{2} - \frac{1}{4} \sqrt{2})/4$ by construction. Hence, by Lemma \ref{BddSize}, we conclude that $\size (\B_i)\geq 2^{|A|}\geq 2^{|R|/4}\geq 2^{n(\frac{1}{2} - \frac{1}{4} \sqrt{2})/4}$
for sufficiently large $n$. 

\end{itemize}

\end{proof}

\section{Conclusions}\label{conclusion}

This  paper improved an earlier result in which the use of
the OBDD proof system is restricted, in a way that the  proof must follow the structure of a given formula. 
We have shown that the OBDD proof system containing two rules,
{\it axiom} and  {\it join}, has lower bounds exponential in $n$  on
refutations for the pigeonhole formulas. On the other hand,  it has been shown in
\cite{CZ2009} that OBDD refutations of the same formulas can be given of 
polynomial size if the { \it projection} rule is  added to  the above two rules. 
Therefore, the result presented in this paper implies that the {\it projection}
rule is responsible for the gap between polynomial and exponential, just like 
the rule in extended resolution is responsible for a similar gap.

\bibliographystyle{eptcs} 

\bibliography{Tveretina}

\begin{thebibliography}{10}
\providecommand{\bibitemstart}[1]{\bibitem{#1}}
\providecommand{\bibitemend}{}
\providecommand{\bibliographystart}{}
\providecommand{\bibliographyend}{}
\providecommand{\url}[1]{\texttt{#1}}
\providecommand{\urlprefix}{Available at }
\providecommand{\bibinfo}[2]{#2}
\bibliographystart

\bibitemstart{AKV2004}
\bibinfo{author}{A.~Atserias}, \bibinfo{author}{P.~Kolaitis} \&
  \bibinfo{author}{M.~Vardi} (\bibinfo{year}{2004}):
  \emph{\bibinfo{title}{Constraint Propagation as a Proof System}}.
\newblock In: {\sl \bibinfo{booktitle}{Principles and Practice of Constraint
  Programming – CP 2004}}, {\sl \bibinfo{series}{LNCS}}
  \bibinfo{volume}{3258}. pp. \bibinfo{pages}{77--91}.
\bibitemend

\bibitemstart{B1986}
\bibinfo{author}{R.~Bryant} (\bibinfo{year}{1986}):
  \emph{\bibinfo{title}{Graph-based algorithms for Boolean function
  manipulation}}.
\newblock {\sl \bibinfo{journal}{IEEE Transactions on Computers}}
  \bibinfo{volume}{8}(\bibinfo{number}{C-35}), pp. \bibinfo{pages}{677--691}.
\bibitemend

\bibitemstart{CZ2009}
\bibinfo{author}{W.~Ch\'{e}n} \& \bibinfo{author}{W.~Zhang}
  (\bibinfo{year}{2009}): \emph{\bibinfo{title}{A direct construction of
  polynomial-size {OBDD} proof of pigeon hole problem}}.
\newblock {\sl \bibinfo{journal}{Information Processing Letters}}
  \bibinfo{volume}{109}(\bibinfo{number}{10}), pp. \bibinfo{pages}{472--477}.
\bibitemend

\bibitemstart{C1976}
\bibinfo{author}{S.~Cook} (\bibinfo{year}{1976}): \emph{\bibinfo{title}{A short
  proof of the pigeon hole principle using extended resolution}}.
\newblock {\sl \bibinfo{journal}{ACM SIGACT News}}
  \bibinfo{volume}{8}(\bibinfo{number}{4}), pp. \bibinfo{pages}{28--32}.
\bibitemend

\bibitemstart{CR1979}
\bibinfo{author}{S.~Cook} \& \bibinfo{author}{R.~Reckhow}
  (\bibinfo{year}{1979}): \emph{\bibinfo{title}{The Relative Efficiency of
  Propositional Proof Systems}}.
\newblock {\sl \bibinfo{journal}{Journal of Symbolic Logic}}
  \bibinfo{volume}{44}(\bibinfo{number}{1}), pp. \bibinfo{pages}{36--50}.
\bibitemend

\bibitemstart{GZ2003}
\bibinfo{author}{J.~F. Groote} \& \bibinfo{author}{H.~Zantema}
  (\bibinfo{year}{2003}): \emph{\bibinfo{title}{Resolution and binary decision
  diagrams cannot simulate each other polynomially}}.
\newblock {\sl \bibinfo{journal}{Discrete Applied Mathematics}}
  \bibinfo{volume}{130}, pp. \bibinfo{pages}{157--171}.
\bibitemend

\bibitemstart{H1985}
\bibinfo{author}{A.~Haken} (\bibinfo{year}{1985}): \emph{\bibinfo{title}{The
  Intractability of Resolution}}.
\newblock {\sl \bibinfo{journal}{Theoretical Computer Science}}
  \bibinfo{volume}{39}, pp. \bibinfo{pages}{297--308}.
\bibitemend

\bibitemstart{K2008}
\bibinfo{author}{J.~Kraj\'{i}\v{c}ek} (\bibinfo{year}{2008}):
  \emph{\bibinfo{title}{An exponential lower bound for a constraint propagation
  proof system based on ordered binary decision diagrams}}.
\newblock {\sl \bibinfo{journal}{The Journal of Symbolic Logic}}
  \bibinfo{volume}{73}(\bibinfo{number}{1}), pp. \bibinfo{pages}{227--237}.
\bibitemend

\bibitemstart{MT1998}
\bibinfo{author}{Ch. Meinel} \& \bibinfo{author}{T.~Theobald}
  (\bibinfo{year}{1998}): \emph{\bibinfo{title}{Algorithms and Data Structures
  in VLSI-Design: OBDD - Foundations and Applications}}.
\newblock \bibinfo{publisher}{Springer-Verlag}, \bibinfo{address}{Berlin,
  Heidelberg, New York}.
\bibitemend

\bibitemstart{SW93}
\bibinfo{author}{D.~Sieling} \& \bibinfo{author}{I.~Wegener}
  (\bibinfo{year}{1993}): \emph{\bibinfo{title}{{NC}-Algorithms for Operations
  on Binary Decision Diagrams}}.
\newblock {\sl \bibinfo{journal}{Parallel Processing Letters}}
  \bibinfo{volume}{3}, pp. \bibinfo{pages}{3--12}.
\bibitemend

\bibitemstart{SA1996}
\bibinfo{author}{R.~Sinnamon} \& \bibinfo{author}{J.~Andrews}
  (\bibinfo{year}{1996}): \emph{\bibinfo{title}{Fault tree analysis and binary
  decision diagrams}}.
\newblock In: {\sl \bibinfo{booktitle}{International Symposium on Product
  Quality and Integrity}}. pp. \bibinfo{pages}{215--222}.
\bibitemend

\bibitemstart{W2000}
\bibinfo{author}{I.~Wegener} (\bibinfo{year}{2000}):
  \emph{\bibinfo{title}{Branching programs and binary decision diagrams: theory
  and applications}}.
\newblock \bibinfo{publisher}{Society for Industrial and Applied Mathematics},
  \bibinfo{address}{Philadelphia, PA, USA}.
\bibitemend

\bibliographyend
\end{thebibliography}

\end{document}